\documentclass[11pt]{paper}

\usepackage{mathrsfs}
\usepackage{amsmath,amssymb,amsthm}
\usepackage{cite}
\usepackage{color}
\usepackage{hyperref}
\usepackage{tikz,pgf}
\usepackage{cancel}
\usepackage[mathscr]{euscript}  
\usepackage{mathtools} 
\numberwithin{equation}{section}

\usepackage{amsmath,amssymb}
\usepackage{hyperref}

\newtheorem{theorem}{\bf Theorem}[section]

\newtheorem{corollary}{\bf Corollary}[section]
\newtheorem{proposition}{\bf Proposition}[section]
\numberwithin{equation}{section}

\newcommand{\bse}{\begin{subequations}}
\newcommand{\ese}{\end{subequations}}

\def\undertilde#1{\mathord{\vtop{\ialign{##\crcr
$\hfil\displaystyle{#1}\hfil$\crcr\noalign{\kern1.5pt\nointerlineskip}
$\hfil\widetilde{}\hfil$\crcr\noalign{\kern-6.5pt}}}}}
\def\underhat#1{\mathord{\vtop{\ialign{##\crcr
$\hfil\displaystyle{#1}\hfil$\crcr\noalign{\kern1.5pt\nointerlineskip}
$\hfil\widehat{}\hfil$\crcr\noalign{\kern-6.5pt}}}}}
\def\underbar#1{\mathord{\vtop{\ialign{##\crcr
$\hfil\displaystyle{#1}\hfil$\crcr\noalign{\kern1.5pt\nointerlineskip}
$\hfil\bar{}\hfil$\crcr\noalign{\kern-6.5pt}}}}}

\newcommand{\be}{\begin{equation}}
\newcommand{\ee}{\end{equation}}

\newcommand{\bea}{\begin{eqnarray}}
\newcommand{\eea}{\end{eqnarray}}

\newcommand{\ddint}{\iint_{D} d\zeta(l,l^\prime)}

\title{Lagrangian 3-form structure for the Darboux system and the KP hierarchy} 

\author{Frank W Nijhoff$^\dagger$ \\
$\dagger$ {\small  School of Mathematics, University of Leeds, Leeds LS2 9JT, UK }\\
\small \texttt{F.W.Nijhoff@leeds.ac.uk}
}

\begin{document}

\maketitle

\begin{abstract}
A Lagrangian multiform structure is established for a generalisation of the Darboux system 
describing orthogonal curvilinear coordinate systems. It has been shown in the past that 
this system of coupled PDEs is in fact an encoding of the entire Kadomtsev-Petviashvili 
(KP) hierarchy in terms so-called Miwa variables. Thus, in providing  
a Lagrangian description of this multidimensionally consistent system amounts to a 
new Lagrangian 3-form structure for the continuous KP system. A generalisation to the 
matrix (also known as non-Abelian) KP system is discussed.

\end{abstract}

\section{Introduction}\label{S:Intro}

The notion of Lagrangian multiforms was introduced in \cite{LobbNij2009} to provide a variational formalism 
for systems integrable in the sense of multidimensional consistency (MDC). This novel variational approach 
to integrable systems allows for the derivation of an entire system (called a \textit{hierarchy}) of 
simultaneous compatible equations from a single variational framework, in which the conventional Lagrange 
function is replaced by a Lagrangian $d$-form integrated over arbitrary hypersurfaces in a space of independent 
variables of arbitrary dimension. Lagrangian multiform theory has undergone a significant development in the last 
decade, (cf. e.g. \cite{Verm,Sleigh-thesis}, or \cite{HJN16} and references therein). It 
has become evident that Lagrangian multiform 
(  or, in its variant formulation,  
\textit{pluri-Lagrangian systems}, cf. e.g. \cite{BobSur15,Suris2012,BolPetSur2013b,BolPetSur2015,Suris2016,SurVerm} ) form a 
universal variational aspect of integrability. It distinguished itself from the conventional least-action principle 
in that, where the latter produces through the standard Euler-Lagrange (EL) equations only 
one equation per component of the field variable, the multiform EL equations comprise a multitude 
of compatible equations for every component of the fields. Furthermore, the Lagrangian components themselves 
have to be very special (they have to be `admissible', which implies `integrable'), and in a precise sense the 
Lagrangians themselves can be considered as solutions of the systems of generalised EL equations.   
   
In this note I will focus on the Darboux system of equations, \cite{Darboux}, which in the original 
notation of Darboux reads 
\begin{equation}\label{eq:Darboux}
\frac{\partial \beta_{kk'}}{\partial \rho_{k''}}= \beta_{kk''}\beta_{k''k'}\  , \quad 
\end{equation}
where the indices $k,k',k''$ run over a set of integers, and the quantities $\beta_{kk'}$, etc., are 
functions of a set of coordinates $\rho_1,\cdots,\rho_n$. These equations describe conjugate 
nets for a system of curvilinear orthogonal coordinates, following on from earlier work by Lam\'e, \cite{Lame}. 
It is well-known that the set of equations 
\eqref{eq:Darboux}, or generalisations thereof, are closely related to integrable three-dimensional 
equations, cf. e.g. \cite{Zakharov96,Doliwa-etal,Doliwa}, in particular the $N$-component wave equation. In fact, in 
\cite{MartinezKonop} it was shown that 
they form a realisation of the KP hierarchy in terms of so-called \textit{Miwa variables}, \cite{Miw82}, which 
are variables depending on a continuous parameter associated with an underlying lattice structure. 
Here, I will show that this set of equations possesses a Lagrangian 3-form structure, 
in the sense of \cite{LNQ,SNC21},   cf. also \cite{BolPetSur2015} . Whereas our previous treatment of the Lagrange 
multiform structure of the continuous KP hierarchy used a representation in terms of pseudo-differential 
operators, going back to \cite{DickeyKP, Dickey-book}, the multiform structure of the  
Darboux system is more compact, and can be viewed as a generating system for the KP hierarchy,
encoding the latter in a more covariant way. 

In the next section I will present this 3-form structure and demonstrate the salient 
multiform features, while in the ensuing section I will discuss the connection to the KP 
hierarchy, and further generalisations in the remainder. Some speculative applications are discussed in the 
Conclusion section.

\section{Lagrangian 3-form structure for the generalised Darboux system} 

The generalised Darboux system reads
\bse\label{eq:Beqs}\begin{align}
&\frac{\partial B_{qr}}{\partial\xi_p}=B_{qp}B_{pr}\  , \quad 
\frac{\partial B_{rq}}{\partial\xi_p}=B_{rp}B_{pq}\  ,\label{eq:Beqsa} \\
&\frac{\partial B_{pr}}{\partial\xi_q}=B_{pq}B_{qr}\  , \quad 
\frac{\partial B_{rp}}{\partial\xi_q}=B_{rq}B_{qp}\  , \label{eq:Beqsb}\\
&\frac{\partial B_{pq}}{\partial\xi_r}=B_{pr}B_{rq}\  , \quad 
\frac{\partial B_{qp}}{\partial\xi_r}=B_{qr}B_{rp}\  , \label{eq:Beqsc}
\end{align} \ese
where the $B_{pq}$, etc., are scalar functions 
(but can be readily generalised to matrices) of the independent variables 
$\xi_p$, $\xi_q$ and $\xi_r$, which are continuous variables labelled by  
parameters $p$, $q$ and $r$ respectively, which themselves are in principle continuous 
variables taking 
values in a continuous subset of the real or complex numbers (hence, the term `generalised').
We assume that these parameters are distinct, and we will not consider for now 
quantities $B$ for which they coincide (quantities of the type $B_{pp}$).  A main property of the system 
\eqref{eq:Beqs} is that it can be extended in a consistent way to an arbitrarily large set of
copies of these equations in terms of additional variables $\xi_s$, etc. similarly labelled 
by values of the parameters. This compatibility is expressed as follows.

\begin{theorem} 
The PDE system \eqref{eq:Beqs} for the quantities $B_{\cdot\cdot}$ is multidimensionally consistent.  
\end{theorem}
\begin{proof}
The proof is by direct computation, introducing a fourth variable $\xi_s$ and 
associated lattice direction with parameter $s$, such that the system 
of independent variables is extended to include $B_{ps},B_{qs},B_{rs}$  and 
$B_{sp},B_{sq},B_{sr}$ obeying relations of the form 
\[ \frac{\partial B_{ps}}{\partial\xi_q}=B_{pq}B_{qs}\  , \quad  \] 
etc. and where the other variables depend also on $\xi_s$ such that  
\[ \frac{\partial B_{pq}}{\partial\xi_s}=B_{ps}B_{sq}\  , \quad  \] 
etc. We then establish by direct computation from the extended system of 
equations comprising \eqref{eq:Beqs} and the PDEs w.r.t. $\xi_s$, the relation
\[ \frac{\partial}{\partial\xi_s}\left(\frac{\partial}{\partial\xi_p}B_{qr} \right) = \frac{\partial}{\partial\xi_p}\left(\frac{\partial}{\partial\xi_s}B_{qr} \right)\  , \] 
by direct computation. Similarly all relations obtained from cross-differentiation 
hold by the same token. 
\end{proof} 

The system \eqref{eq:Beqs} possesses a \textit{Lax multiplet}, cf. e.g. \cite{Doliwa}, in the following sense. 
\begin{proposition}
 
If the system \eqref{eq:Beqs} is satisfied each of the following linear overdetermined systems 
(one system for a parameter-labelled family of functions $\Phi_\cdot$ and another system for the 
parameter-family of functions $\Psi_\cdot$)
\begin{equation}\label{eq:Lax}
\frac{\partial \Phi_{q}}{\partial\xi_p}=B_{qp}\Phi_{p}\  , \quad {\rm and}\quad 
\frac{\partial \Psi_{r}}{\partial\xi_p}=\Psi_{p}B_{pr}\  ,
\end{equation} 
respectively, (and similar relations for all variables $\xi_q$ and $\xi_r$) is \textit{consistent} in the 
sense of possessing a common general solution.  
 
\end{proposition} 
\begin{proof}
Again, this is by direct computation. Cross-differentiation of two copies of the first equation 
of \eqref{eq:Lax} we get the equality
\begin{align*}
&\frac{\partial}{\partial\xi_r}\left(\frac{\partial \Phi_{q}}{\partial\xi_p}\right)=
\left(\frac{\partial}{\partial\xi_r}B_{qp}\right)\Phi_{p}+ B_{qp}B_{pr}\Phi_r\ \\ 
& = \frac{\partial}{\partial\xi_p}\left(\frac{\partial \Phi_{q}}{\partial\xi_r}\right)=
\left(\frac{\partial}{\partial\xi_p}B_{qr}\right)\Phi_{r}+ B_{qr}B_{rp}\Phi_p \  ,
\end{align*} 
and hence the equality of the coefficients of $\Phi_r$ and $\Phi_p$ give us the desired differential 
equations for $B_{qp}$ and $B_{qr}$ respectively. The same hold true for the `adjoint' Lax 
multiplet in terms of the functions $\Psi_\cdot$.
\end{proof} 

We note that the Lax multiplets \eqref{eq:Lax} can be obtained from the Darboux 
system itself, relying on the multidimensional consistency, by identifying the 
Lax wave functions $\Phi$ and $\Psi$ by fixing two, possibly separate, 
directions in the space of independent 
variables, $\xi_k$ and $\xi_l$, say (where $k$ and $l$ play the role of spectral parameters), such that $\Phi_p=B_{pk}$ and $\Psi_p=B_{lp}$.   
Furthermore, the quantities $\Phi$ and $\Psi$ obey a linear homogeneous set of equations of the 
form 
\bse\label{eq:homLax} 
\begin{align}
& \partial_p\partial_q\Phi_r= (\partial_p\ln \Phi_q)\partial_q\Phi_r + (\partial_q\ln \Phi_p)\partial_p\Phi_r\ , \\  
& \partial_p\partial_q\Psi_r= (\partial_p\ln \Psi_q)\partial_q\Psi_r + (\partial_q\ln \Psi_p)\partial_p\Psi_r\ , 
\end{align}\ese 
thus obeying both an identical equation. 

We now introduce the Lagrangian structure. Let us consider the following 
Lagrangian components 
\begin{align}\label{eq:Lagr}
\mathcal{L}_{pqr} = & \tfrac{1}{2}\left(B_{rq}\partial_{\xi_p}B_{qr}-B_{qr}\partial_{\xi_p}B_{rq} \right)  
  +\tfrac{1}{2}\left(B_{qp}\partial_{\xi_r}B_{pq}-B_{pq}\partial_{\xi_r}B_{qp} \right) \nonumber \\ 
   & +\tfrac{1}{2}\left(B_{pr}\partial_{\xi_q}B_{rp}-B_{rp}\partial_{\xi_q}B_{pr} \right)  
   + B_{rp} B_{pq} B_{qr} - B_{rq}B_{qp} B_{pr}  \  . 
\end{align} 

\noindent Then we have the following main statement 

\begin{theorem} 
The differential of the Lagrangian 3-form 
\begin{align}\label{eq:Lagr3form} 
{\sf L}:= &\mathcal{L}_{pqr}\,{\rm d}\xi_p\wedge{\rm d}\xi_q\wedge{\rm d}\xi_r+ 
\mathcal{L}_{qrs}\,{\rm d}\xi_q\wedge{\rm d}\xi_r\wedge{\rm d}\xi_s+ \nonumber \\ 
& \quad +\mathcal{L}_{rsp}\,{\rm d}\xi_r\wedge{\rm d}\xi_s\wedge{\rm d}\xi_p+ 
\mathcal{L}_{spq}\,{\rm d}\xi_s\wedge{\rm d}\xi_p\wedge{\rm d}\xi_q\ ,  
\end{align} 
has a `double zero' on the solutions of the set of generalised Darboux equations \eqref{eq:Beqs}, 
i.e. ${\rm d}{\sf L}$ can be written as 
\begin{equation}
{\rm d}\mathsf{L}= \,\mathcal{A}_{pqrs}\,{\rm d}\xi_p\wedge{\rm d}\xi_q\wedge{\rm d}\xi_r\wedge{\rm d}\xi_s
\end{equation}  
with the coefficient $\mathcal{A}_{pqrs}$ being a sum of products of factors which vanish 
on solutions of the EL equations. 
\end{theorem} 
\begin{proof} 
Computing the components of the differential ${\rm d}\mathsf{L}$ we obtain 
\begin{align*} 
& \partial_{\xi_s}\mathcal{L}_{pqr}-\partial_{\xi_p}\mathcal{L}_{qrs} +\partial_{\xi_q}\mathcal{L}_{rsp} 
-\partial_{\xi_r}\mathcal{L}_{spq}= \\ 
& \quad \Gamma_{s;rq}\Gamma_{p;qr} - \Gamma_{p;rq}\Gamma_{s;qr} 
+ \Gamma_{s;qp}\Gamma_{r;pq} - \Gamma_{r;qp}\Gamma_{s;pq} \\ 
&\quad +\Gamma_{s;pr}\Gamma_{q;rp} - \Gamma_{q;pr}\Gamma_{s;rp} 
+\Gamma_{q;sr}\Gamma_{p;rs} - \Gamma_{p;sr}\Gamma_{q;rs} \\ 
&\quad +\Gamma_{p;sq}\Gamma_{r;qs} - \Gamma_{r;sq}\Gamma_{p;qs} 
+\Gamma_{q;ps}\Gamma_{r;sp} - \Gamma_{r;ps}\Gamma_{q;sp}\ ,  
\end{align*} 

where 
\[ \Gamma_{p;qs}=\partial_{\xi_p}B_{qs}-B_{qp}B_{ps}\  ,  \]
and similarly for the other indices.  The set of generalised EL equations in this case 
are obtained from $\delta\mathcal{A}_{pqrs}=0$, repeating the general argument, cf. e.g. \cite{SurVerm,SNC20,Sleigh-thesis}, for deriving the EL equations from the differential 
of the Lagrangian multiform. Thus, since all the variations $\delta B_{pq}$ etc. 
and their first derivatives,  
are independent, the coefficients are precisely all the combinations 
$\Gamma_{r;pq}$, etc. which will have to vanish at the critical point for 
the action 
\begin{equation}\label{eq:action}
{\sf S}[\mathbf{B}(\boldsymbol{\xi});\mathcal{V}]=\int_{\mathcal{V}} {\sf L}= 
\int_{ \mathcal{W} }\,{\rm d}{\sf L}\  ,  
\end{equation}
integrated over any arbitrary 3-dimensional closed hypersurfaces $\mathcal{V}$ in the multivariable space of all the 
$\xi_p$'s,   and where the enclosed volume 
$\mathcal{W}$ is such that $\mathcal{V}=\partial\mathcal{W}$. 
\end{proof}

As a corollary the statement of Theorem 2.2 holds more generally for Lagrangian 3-forms 
embedded in a higher-dimensional space of independent variables, namely 

\begin{corollary}  
In a space of variables $\{ \boldsymbol{p}=(p_j)_{j\in I} \}$ where the $p_i$ denote complex valued 
continuous variables labelled by an index set $I$, the Lagrangian 3-form 
\begin{align} \label{eq:gen3form} 
\mathsf{L}= \sum_{i,j,k\in I} \mathcal{L}_{p_i,p_j,p_k}\,{\rm d}\xi_{p_i}\wedge {\rm d}\xi_{p_j} 
\wedge{\rm d}\xi_{p_k}\ , 
\end{align} 
with $\mathcal{L}_{p_i,p_j,p_k}$ as given in \eqref{eq:Lagr}, we have that 
\[ {\rm d}\mathsf{L}= \sum_{i,j,k,l\in I} \mathcal{A}_{p_i,p_j,p_k,p_l}\,
{\rm d}\xi_{p_i}\wedge {\rm d}\xi_{p_j}\wedge{\rm d}\xi_{p_k}\wedge{\rm d}\xi_{p_l}\ , \] 
has a double zero on solutions of the system \eqref{eq:Beqs} written in the relevant variables 
labelled by $p_i,p_j,p_k,p_l$. 
\end{corollary} 
\noindent The proof is an obvious extension of the one for Theorem 2.2, assuming that all labels 
of the $p_{i_\nu}$ are distinct.

The variational equations obtained from  $\delta{\rm d}{\sf L}=0$ 
for the Lagrangian multiform \eqref{eq:Lagr3form} 
constitute the set of multiform Euler-Lagrange equations in the language  
of the variational bicomplex, cf. e.g. \cite{Dickey-book}. In fact, considering all the fields 
$B_{\cdot\cdot}$ independent for different labels, this compact form of the variational 
equations, through the double-zero form, implies the vanishing of all the 
factors $\Gamma$ in the above computation, which amounts to the set of generalised Darboux 
equations. Furthermore, as a direct consequence of the double-zero form for ${\rm d}{\sf L}$ the 
Lagrangian multiform ${\sf L}$ is closed on solutions of the Darboux system 
\eqref{eq:Beqs}, i.e. $\left.{\rm d}{\sf L}\right|_{\rm EL}=0$ 
(but not trivially so, only `on-shell'), which implies that 
for the critical fields which obey the Darboux system the action is invariant 
under smooth deformations of the hypersurface $\mathcal{V}$. This is precisely the 
phenomenon of multidimensional consistency: the Darboux system is compatible 
on any hypersurface in the multidimensional space of Miwa variables.  

It is important to mention in this context that multi-time Euler-Lagrange equations 
were derived in various papers, notably \cite{Y-KLN,Suris2012} in the case of Lagrangian 1-forms, 
\cite{Suris2016,SurVerm}, continuous 2-forms (and in \cite{BolPetSur2013b,LobbNij2018} in the 
discrete case) and, more generally, in \cite{Verm,Sleigh-thesis} in the general continuous case. 
These derivations follow different approaches, but they have in common that the conventional 
EL equations associated with a specific choice (meaning in the present context fixing $\mathcal{V}$ 
in \eqref{eq:action}) of integration manifold (in the 
multi-time space of independent variables) must hold simultaneously for all possible choices of 
$\mathcal{V}$, and are supplemented by a set of additional constraints linking the Lagrangian 
components arising from `alien derivatives' (i.e. derivatives w.r.t. independent variables 
that are not integrated over in the action functional 
${\sf S}[\mathbf{B}(\boldsymbol{\xi});\mathcal{V}]$ 
along the direction of the components in question of the Lagrangian multiform $\mathsf{L}$). 
I don't intend to write down the general formulae here, as they require a lot of additional 
notations, cf. e.g. \cite{SurVerm,Verm,Sleigh-thesis} for a detailed presentation. 
In fact, in \cite{Sleigh-thesis} the double-zero condition 
is shown to be a sufficient condition for the multi-time EL equations to be satisfied,  
and this is all we need for the purpose of the present paper.

It may also need pointing out that in the case that the dimension of the embedding space of independent variables coincides with 
the number of variables in the system \eqref{eq:Beqs}, i.e. 
when the cardinality of the index set $|I|=3$ in \eqref{eq:gen3form},  we recover the conventional case of 
a Lagrangian volume-form with Lagrangian density $\mathcal{L}_{pqr}= \mathcal{L}_{p_i,p_j,p_k}$, in which case the standard Euler-Lagrange equations yield
\begin{align}\label{eq:standL}
\frac{\delta\mathcal{L}_{pqr}}{\delta B_{pq}}=-\frac{\partial B_{qp}}{\partial \xi_r}+B_{qr}B_{rp} \ , 
\quad 
\frac{\delta\mathcal{L}_{pqr}}{\delta B_{qp}}=\frac{\partial B_{pq}}{\partial \xi_r}+B_{pr}B_{rq} \ ,
\end{align}  
and similarly for the components $B_{qr}$, $B_{rq}$, $B_{pr}$ and $B_{rp}$. Thus, we obtain the 
Darboux system \eqref{eq:Beqs}, but the integrability in the sense of MDC is not evident from the 
conventional Lagrangian formalism.

Another corollary of the multiform structure is that we also have a variational 
description of the Lax system \eqref{eq:Lax}. To see this, note that we can extend the set of Miwa variables to include variables $\xi_k$ 
associated with a `spectral parameter' $k$. Thus, we are led to the following statement. 

\begin{corollary}
The Lagrangian 3-form 
\begin{align}\label{eq:LaxLagr3form} 
{\sf L}_{(k)}:= &\mathcal{L}_{pq(k)}\,{\rm d}\xi_p\wedge{\rm d}\xi_q\wedge{\rm d}\xi_k+ 
\mathcal{L}_{qr(k)}\,{\rm d}\xi_q\wedge{\rm d}\xi_r\wedge{\rm d}\xi_k+ \nonumber \\ 
& \quad +\mathcal{L}_{rp(k)}\,{\rm d}\xi_r\wedge{\rm d}\xi_p\wedge{\rm d}\xi_k+ 
\mathcal{L}_{pqr}\,{\rm d}\xi_p\wedge{\rm d}\xi_q\wedge{\rm d}\xi_r\ ,  
\end{align} 
with components   
\begin{align}\label{eq:LaxLagr}
\mathcal{L}_{pq(k)} = & \tfrac{1}{2}\left(\Psi_{q}\partial_{\xi_p}\Phi_{q}-(\partial_{\xi_p}\Psi_{q})\Phi_q \right)  
-\tfrac{1}{2}\left(\Psi_{p}\partial_{\xi_q}\Phi_{p}-(\partial_{\xi_q}\Psi_{p})\Phi_p \right) \nonumber \\ 
& +\tfrac{1}{2}\left(B_{qp}\partial_{\xi_k}B_{pq}-B_{pq}\partial_{\xi_k}B_{qp} \right)  
   + \Psi_{p} B_{pq} \Phi_{q} - \Psi_{q}B_{qp} \Phi_{p}  \  , 
\end{align} 
through the EL equations $\delta{\rm d}\mathsf{L}_{(k)}=0$ constitutes a variational description of the 
Lax multiplet \eqref{eq:Lax}.  
\end{corollary} 

A similar variational description of the Lax system was obtained 
in \cite{SNC19} for the 1+1-dimensional Lax system associated with the 
so-called Zakharov-Mikhailov action. We note that the Lagrangian 3-form \eqref{eq:LaxLagr3form} 
should really be considered as a Lagrangian 2-form when integrating out the direction $\xi_k$ 
associated with the (fixed) spectral variable. 

\section{Discrete Darboux system} 

A discrete analogue of the Darboux system of orthogonal coordinate systems, i.e. 
was found in  \cite{BogdKonopel95,DS97}  where its integrability was inserted. In fact, an 
interesting connection with integrable quadrilateral lattices was discovered, as 
well as with multidimensional circular lattices, cf. \cite{CDS97,MDS97}.  
The corresponding discrete analogue of the generalised Darboux system \eqref{eq:Beqs} reads
\bse\label{eq:dBeqs}\begin{align}
& \Delta_p B_{qr}=B_{qp}T_pB_{pr}\  , \quad 
\Delta_p B_{rq}=B_{rp}T_pB_{pq}\  , \label{eq:dBeqsa} \\
&\Delta_q B_{rp}=B_{rq}T_qB_{qp}\  ,
\quad \Delta_q B_{pr}=B_{pq}T_qB_{qr}\  , \label{eq:dBeqsb}\\
&\Delta_r B_{pq}=B_{pr}T_rB_{rq}\  , \quad 
\Delta_r B_{qp}=B_{qr}T_rB_{rp}\  , \label{eq:dBeqsc}
\end{align} \ese
where $\Delta_p$ denotes the difference operator $\Delta_p=T_p-{\rm id}$, and 
where $T_p$ is a shift operator in a discrete variable $n$ associated with the 
`lattice parameter' $p$ (see next section for the relation between $p$ and $n$). 
This system is related to other 
multidimensional lattice systems that were formulated in \cite{Nij85LMP}. 
\begin{theorem}
The system of difference equations \eqref{eq:dBeqs} is multidimensionally consistent, and furthermore, it is consistent 
with the differential system \eqref{eq:Beqs}.
\end{theorem}
\begin{proof}
The consistency of the set of difference equations is by direct computation. For instance, rewriting the difference equation as follows
\begin{align*} 
& B_{qr}= T_pB_{qr}-B_{qp}T_pB_{pr}= \\ 
& = T_p\left(T_sB_{qr}-B_{qs}T_sB_{sr}\right)-
B_{qp}T_p\left(T_sB_{pr}-B_{ps}T_sB_{sr}\right) \\ 
& =T_pT_sB_{qr}-(T_pB_{qs})T_pT_sB_{sr}-B_{qp}T_pT_sB_{pr} 
+B_{qp}(T_pB_{ps})T_pT_sB_{sr}  
\end{align*} 
which is equal to the same expression with the labels $p$ and $s$ interchanged.  Thus, the latter is equal to 
\[ =T_sT_pB_{qr}-(T_sB_{qp})T_sT_pB_{pr}-B_{qs}T_sT_pB_{sr} 
+B_{qs}(T_sB_{sp})T_sT_pB_{pr}\ .  \]
Assuming that the shifts $T_p$ and $T_s$ commute, and collecting the factors with $T_pT_sB_{pr}$ and the ones with 
$T_pT_sB_{pr}$, regarding the latter as independent, we obtain the relations 
\[ T_pB_{qs}-B_{qs}-B_{qp}T_pB_{ps}=0 \ , \quad {\rm and}\quad T_sB_{qp}-B_{qp}-B_{qs}T_sB_{sp}=0\  ,  \]  
which are two of the discrete Darboux equations. Thus, the relations are consistent under mutual shifts. The compatibility with the continuous Darboux system \eqref{eq:Beqs} follows from a 
similar computation. Abbreviating $\partial/\partial\xi_p$ by $\partial_p$ we get 
\begin{align*} 
& \partial_p (T_sB_{qr})= \partial_p\left(B_{qr}+B_{qs}T_sB_{sr} \right)=\partial_pB_{qr}+(\partial_pB_{qs})T_sB_{sr}+ B_{qs}\partial_pT_sB_{sr}\\ 
& {\rm whereas} \\ 
& T_s\partial_p B_{qr} = T_s(B_{qp}B_{pr}) =(T_sB_{qp})T_sB_{pr}=(B_{qp}+B_{qs}T_sB_{sp})T_sB_{pr} \\ 
&\Rightarrow \quad  B_{qp}B_{pr}+B_{qp}B_{ps}T_sB_{sr}+\cancel{B_{qs}T_s\left(B_{sp}B_{pr}\right)} \\ 
&\qquad\quad =B_{qp} T_sB_{pr}+\cancel{B_{qs}(T_sB_{sp})T_sB_{pr}}\ , 
\end{align*} 
and the remaining terms cancel as well due to the discrete Darboux relation. 

\end{proof}

Similarly to the continuous case we have a Lax system, and its adjoint, given by 
\begin{equation}
\Delta_p\Phi_q=B_{qp}T_p\Phi_p\ , \quad \Delta_p\Psi_q=\Psi_p T_pB_{pq}\  , 
\end{equation} 
and the homogeneous linear difference system for an eigenfunctions $\Phi_r$, $\Psi_r$ respectively,   
\begin{align}
& \Delta_p\Delta_q \Phi_r =\frac{\Delta_p(T_q\Phi_q)}{T_q\Phi_q}\,\Delta_q\Phi_r 
+\frac{\Delta_q(T_p\Phi_p)}{T_p\Phi_p}\, \Delta_p\Phi_r\ ,     \\ 
&  \Delta_p\Delta_q \Psi_r =\frac{\Delta_p\Psi_q}{T_p\Psi_q}\,\Delta_q(T_p\Psi_r) 
+\frac{\Delta_q\Psi_p}{T_q\Psi_p}\,\Delta_p(T_q\Psi_r)\ . 
\end{align}
Note that in the discrete case the equations for the eigenfunction and its adjoint are no 
longer the same. It is natural to assume that the discrete Darboux system \eqref{eq:dBeqs}, like its 
continuous counterpart \eqref{eq:Beqs}, admits a Lagrangian 3-form structure. I intend to settle 
this question in a future publication \cite{Nijtbp}.

\section{Connection with the (scalar) KP system} 

The KP system of equations is often 
introduced as the set of Lax equations arising from a Lax operator in a ring of 
pseudo-differential operators with respect to a singled-out variable $x$, cf. 
\cite{Sato}. This 
has a disadvantage that the inherent covariant structure of the KP system is 
broken, and not all independent variables (the higher time variables) appear on the 
same footing as the variable $x$. A more covariant approach is provided by the 
`direct linearisation' set-up, cf. e.g. \cite{FN17} and references therein, 
where there is no need to single out a particular variable to describe the KP 
hierarchy. It can be argued that the generalised Darboux 
system \eqref{eq:Beqs} provides also a covariant description but in the sense of 
encoding the hierarchy through Miwa variables, cf. \cite{MartinezKonop}. In that 
sense the generalised Darboux system is similar in spirit as the `hierarchy generating PDEs' of 
\cite{NHJ,TN}, but for a 3-dimensional system of PDEs instead of the KdV or 
Boussinesq hierarchies respectively. 

Solutions of the discrete KP system were considered in \cite{NCWQ84,NC90,FN16} using the \textit{direct linearisation} (DL)  
approach, cf. also \cite{SAF1984}. The dynamics is governed 
by \textit{plane-wave factors} which take the form  
\bse\label{eq:freeref}\begin{align}
& \rho_k= \left[\prod_\nu (p_\nu-k)^{n_\nu}\right]\,\exp\left\{k\xi-\sum_\nu \frac{\xi_{p_\nu}}{p_\nu-k}  \right\}\ , \\  
& \sigma_{k'}= \left[\prod_\nu (p_\nu-k')^{-n_\nu}\right]\,\exp\left\{-k'\xi+\sum_\nu \frac{\xi_{p_\nu}}{p_\nu-k'}  \right\}\ .
\end{align}
\ese 
Here the $\xi_{p_\nu}$ are the independent variables of the generalised 
Darboux system, and the $n_\nu$ are associated discrete variables, in 
terms of which the KP $\tau$ function obeys the compatible set Hirota bilinear equations\footnote{In the 
literature, cf. e.g. \cite{Doliwa-etal,MartinezKonop}, the dependence on discrete shifts in what is 
essentially the \emph{scalar} KP system, is confusingly often referred to as the 
`multi-component KP hierarchy', (because lattice-shifted variables are considered as components).  
This should not be confused with the \emph{matrix KP system}, cf. e.g. \cite{Konop82,Nij85LMP}, 
which in my opinion more rightfully deserves the name `multicomponent', and which is related to the 
system in section 5. The difference between the two resides in that the scalar KP is governed 
essentially by a scalar integral measure in the underlying DL framework, while the matrix KP is 
governed by a matrix measure, and hence has a much richer solution  structure.} 
\begin{equation}\label{eq:Hirota} 
(p-q)(T_pT_q\tau)T_r\tau+ (q-r)(T_qT_r\tau)T_p\tau+(r-p)(T_rT_p\tau)T_q\tau=0\ , 
\end{equation}
where $T_{p_\nu}$ ($p,q,r$ being any three of the $p_\nu$) denotes the elementary shift in the variable $n_\nu$ 
associated with $p_\nu$ (which, in this context, has the interpretation 
of a \textit{lattice parameter} measuring the grid width in the 
discrete direction labelled by $n_\nu$). 

The interplay between discrete 
and continuous variables turns out to be an essential feature of the 
structure. In fact, the $\tau$-function obeys the relations 
\begin{equation}\label{eq:taurel}
\frac{\partial\tau}{\partial\xi_p}=- \left(T_p^{-1}\frac{{\rm d}}{{\rm d}p}T_p\right)\tau:= 
\lim_{\varepsilon\to 0} \frac{T_p^{-1}T_{p-\varepsilon}\tau-\tau}{\varepsilon}\  , 
\end{equation}
for any of the parameters $p_\nu=p$, where we should think of the 
$p-\varepsilon$ as the lattice parameter associated with 
  lattice directions with elementary lattice shift $T_{p-\varepsilon}$ for 
all arbitrary small $\varepsilon$ .
Using the identification between lattice shifts and derivatives as in \eqref{eq:taurel}, we can 
perform a limit $r\to p$ on \eqref{eq:Hirota} and thus obtain the following 
differential-difference equation for $\tau$
\begin{equation}\label{eq:diffdifftau} 
(p-q) \left(\tau\,T_q\frac{\partial\tau}{\partial \xi_p}-(T_q\tau)\,\frac{\partial\tau}{\partial \xi_p}\right) 
=\tau\,T_q\tau - (T_p\tau)\,T_qT_p^{-1}\tau\  . 
\end{equation}  
Furthermore, the $\tau$-function also obeys the differential-difference equation
\begin{equation}\label{eq:taudiff}
1+(p-q)^2 \frac{\partial^2\ln\tau}{\partial \xi_p\,\partial \xi_q}=\frac{(T_pT_q^{-1}\tau)T_qT_p^{-1}\tau}{\tau^2}\ , 
\end{equation} 
which can be readily cast into bilinear form. In fact, eq. \eqref{eq:taudiff} is the bilinear form of the 
2D Toda equation (with the discrete variable along the skew-diagonal lattice direction in the lattice 
generated by the $T_p$ and $T_q$ shifts).    

It turns out that the Darboux variables of the 
system \eqref{eq:Beqs} can be expressed in terms of the KP $\tau$-function exploiting the underlying discrete structure\footnote{In \cite{MartinezKonop,Doliwa-etal} a similar connection was exhibited, which differs 
from the present one in that my presentation is based on results from the DL approach, 
whereas the Sato or fermionic type approach seems to cover a more restricted solution sector of 
the theory.}. To do so 
consider the quantities 
\begin{equation}\label{eq:S} 
S_{a,b}= \frac{T_a^{-1}T_b\tau}{\tau} \  , 
\end{equation}
as a consequence of \eqref{eq:Hirota} and \eqref{eq:diffdifftau} obey the following 
relations 
\bse\label{eq:Srels}
\begin{align} 
& (p-b)T_pS_{a,b}-(p-a)S_{a,b}=(a-b) S_{a,p}T_pS_{p,b}\  , \label{eq:Sdrels} \\ 
&(p-a)(p-b)\frac{\partial S_{a,b}}{\partial\xi_p}= 
(a-b)\left(S_{a,p}S_{p,b}-S_{a,b}\right)\ .  \label{eq:Screls} 
\end{align} 
\ese 
Similar relations appeared in \cite{MartinezKonop,Doliwa-etal} derived from a 
different perspective from \cite{FN16}.  
These relations are compatible for all parameters $p$ and corresponding 
shifts and derivatives w.r.t. the corresponding Miwa variables $\xi_p$. 
They form the basis for the generalised Darboux and a discrete analogue 
of the Darboux system, where the latter can be obtained 
by a gauge transformation with factors $\rho_{-a}\sigma_{b}$ of the 
form \eqref{eq:freeref}. Furthermore, the quantity $S=S_{a,b}$ 
obeys the following three-dimensional partial difference equation, \cite{NCWQ84},  
\begin{align} 
& \frac{\left[(p-b)T_pT_qS-(p-a)T_qS\right]\, \left[(q-b)T_qT_rS-(q-a)T_rS\right]} 
{\left[(p-b)T_pT_rS-(p-a)T_rS\right]\, \left[(q-b)T_pT_qS-(q-a)T_pS\right]} \nonumber \\ 
& \qquad \times\frac{\left[(r-b)T_pT_rS-(r-a)T_pS\right]} {\left[(r-b)T_qT_rS-(r-a)T_qS\right]}=1 
 \end{align} 
which is essentially the lattice Schwarzian KP equation, first given in its well-known 
pure form in \cite{DN91}. 

The KP hierarchy can be obtained by the expansions 
\begin{align}\label{eq:Miwa} 
& t_j=\delta_{j,1}\xi+\sum_{\nu}\left(\frac{\xi_{p_\nu}}{p_\nu^{j+1}} + \frac{1}{j}\,\frac{n_\nu}{p_\nu^j}\right) \nonumber \\ 
& \Rightarrow\quad T_{p_\nu}\tau=\tau\left(\{t_j+\frac{1}{jp_\nu^j}\}\right)\quad {\rm and}\quad 
\frac{\partial\tau}{\partial\xi_{p_\nu}}=\sum_{j=1}^\infty \frac{1}{p_\nu^{j+1}}\frac{\partial\tau }{\partial t_j}\  ,  
\end{align} 
where the $t_j$ are the usual independent time-variables in the hierarchy.

From \eqref{eq:Screls} it follows that the Darboux quantities can be 
identified as 
\begin{equation}\label{eq:B} 
B_{pq}= \frac{\sigma_{p}\rho_{q} S_{p,q}}{q-p}= \sigma_{p}\rho_{q} \frac{T_p^{-1}T_q\tau}{(q-p)\tau} \  , \quad q\neq p\ , 
\end{equation}
from which, together with \eqref{eq:freeref} we get the relations \eqref{eq:Beqs} whenever $q\neq p$. 
When $q=p$, we have $B_{pp}=\mathcal{C}\partial_{\xi_p}(\ln\tau)$, where $\mathcal{C}$ is some constant normalisation factor.  

The eigenfunctions of the Lax multiplet are obtained from 
\begin{equation}\label{eq:eigen} 
\phi_a(k)=\frac{S_{a,k}\rho_k}{a-k}\ , \quad 
\psi_b(k')=\frac{S_{k',b}\sigma_{k'}}{b-k'}\  ,  
\end{equation} 
which obey the set of relations 
\bse\label{eq:ukrels}\begin{align}
& (p-a) \phi_a(k)=T_p\phi_a(k)-S_{a,p}T_p\phi_p(k)\ , \\ 
& (p-b)T_p\psi_b(k')=\psi_b(k')-\psi_{p}(k')T_pS_{p,b}\ , \\ 
& (a-p)\frac{\partial}{\partial\xi_p}\phi_a(k) =\phi_a(k) 
-S_{a,p}\phi_p(k)\ , \\ 
& -(b-p)\frac{\partial}{\partial\xi_p} \psi_b(k') 
= \psi_b(k') -\psi_{p}(k')S_{p,b}\  ,  
\end{align}\ese 
the compatibility conditions of which reproduce eqs. \eqref{eq:Srels}.

Within the setting of the DL approach, the following combination of the quantities $S$, 
for arbitrary values of $c$, possesses a quadratic eigenfunction expansion of the form 
\begin{equation}\label{eq:quadexps}
 S_{a,b}-S_{a,c}S_{c,b} = \ddint (l-l') 
\frac{(c-a)(c-b)}{(c-l)(c-l')}\,\phi_a(l)\psi_{b}(l')\  .    
\end{equation}
Here the integration is over an arbitrary measure in a region $D\subset 
\mathcal{C}\times\mathcal{C}$ of values of $l,l'$ in a spectral space 
with measure ${\rm d}\zeta(l,l')$. Under special conditions 
these integrals correspond to the generalised Cauchy integrals arising 
in the $\bar{\partial}$ problem or nonlocal Riemann-Hilbert problems 
for the KP type spectral problems, cf. \cite{ZM85,Konop}. (The choices of $c$ must 
be such that singularities in the integrals are avoided, and this requires some 
conditions on the integrations, which play a role when we consider special solutions.
We will not address these issues of analysis here.)  
Note that when $c=p_\nu$, i.e. coincides with  any of the parameters 
associated with the Miwa variables $\xi_p$ , then the left hand side 
of \eqref{eq:quadexps} coincides with the expression on the right-hand 
side of \eqref{eq:Srels}. I.o.w. the right-hand side of \eqref{eq:quadexps} provides a quadratic eigenfunction expansion for 
the derivative of $S_{a,b}$ w.r.t. $\xi_p$ (modulo a constant factor).  
However \eqref{eq:quadexps} is independent of the choice of Miwa variables and holds for any $c$. In particular in the limit 
$c\to\infty$ we obtain the following fundamental bilinear identity for the solution of the $\tau$ function associated with the 
choice of measure and integration region $D$:   
\begin{align}\label{eq:fundtau}
& \ddint \frac{(l-l')\rho_l\sigma_{l'}}{(a-l')(a'-l)}(T_{a'}^{-1}T_{l}\tau)(T_{l'}^{-1}T_{a}\tau)= \nonumber \\ 
& \qquad = \tau\,(T_{a'}^{-1}T_{a}\tau)-(T_{a'}^{-1}\tau)(T_{a}\tau)\  , 
\end{align} 
which can be considered as a \textit{bilinear integro-difference equation} for the $\tau$ function.  
The relation \eqref{eq:fundtau} is reminiscent of the fundamental bilinear 
identity that plays a central role in the Sato approach to the KP hierarchy, cf. also 
\cite{BogdKonopel}, which is however not the approach taken here to derive this relation. It is maybe useful to mention 
at this juncture that, while all definitions of a $\tau$-function are in a sense non-universal and depend on the 
solution class under consideration, what may be the most general definition of a $\tau$-function was formulated in   
\cite{Nij-tau85}, namely in terms of a fermionic path integral associated with the direct linearising transform structure.

\section{Generalisation to the matrix case} 

In a talk at the June 1987 NEEDS meeting I presented a 2+1-dimensional Lagrangian matrix 
KP system which effectively amounts to a matrix generalisation of the Darboux system, that became a focus of attention in the mid 1990s.  
We proposed the following Lagrangian, cf. \cite{NijMaill},  
\begin{align}\label{eq:GLagr} 
\mathcal{L}_{ijk} & = \tfrac{1}{2}\textrm{tr}\left\{ G_{ij}J_i(\partial_k G_{ji})J_j-(\partial_k G_{ij})J_i G_{ji}J_j +\textrm{cycl. (ijk)} \right\}  \nonumber \\ 
& \quad - \textrm{tr}\left\{G_{ij} J_i G_{ki} J_k G_{jk} J_j 
-G_{ji} J_j G_{kj} J_k G_{ik} J_i \right\}\ ,  
\end{align} 
which is a matrix generalisation of \eqref{eq:Lagr}. In fact, the $G_{ij}$ are 
$N\times N$ matrix functions of dynamical variables $x_i=\xi^{J_i}_{l_i},x_j=\xi^{J_j}_{l_j},
x_k=\xi^{J_k}_{l_k}, \dots $, which are labelled not only by a continuous parameter $l_{\cdot}$ 
(like the $p,q,r$ in the scalar case), but also by a matrix $J_{\cdot}$ which in a sense 'tunes' 
a hierarchy of associated KP type equations.   
while the $J_i,J_j,J_k$ are constant $N\times N$ 
matrices, which commute among themselves\footnote{In fact, one can also 
consider the non-commutative case $[J_i,J_j]=\Gamma_{ij}^kJ_k$, in which case we get non-commuting flows on 
a loop group, for which a Lagrangian description was proposed recently in \cite{CNSV2022} for 
(1+1)-dimensional systems.}, i.e., $[J_i,J_j]=[J_j,J_k]=[J_k,J_i]=0$. In \eqref{eq:GLagr} 
we have denoted $\partial/\partial \xi_{l_j}=:\partial_j$, etc. for the sake of brevity.  
Like \eqref{eq:Lagr} the Lagrangian \eqref{eq:GLagr} can be viewed as a component of a Lagrangian 3-form 
\begin{equation} \label{eq:matLagr3form}
\mathsf{L}= \sum_{i<j<k} \mathcal{L}_{ijk}\,{\rm d}x_i\wedge {\rm d}x_j\wedge {\rm d}x_k\  , 
\end{equation} 
which is closed on solutions of the Euler Lagrange equations
\begin{equation}\label{eq:Geqs}
\partial_i G_{jk}=G_{ik}J_iG_{ji}\  , \quad i\neq j\neq k\neq i\  . 
\end{equation}

The main statement is that these Lagrangians form the components of a Lagrangian 3-form. Thus, we have  
\begin{theorem}
The Lagrangian 3-form \eqref{eq:matLagr3form}
has a double zero on solutions of the fundamental set of equations \eqref{eq:Geqs}. 
\end{theorem}   
\begin{proof}
The proof is again computational, and in essence similar to the one of Theorem 2.2, with the main difference 
occurring in the matrix ordering within the trace. Computing the differential of $\mathsf{L}$ 
we get in the matrix case 
\[ {\rm d}\mathsf{L}=\sum_{i,j,k,l} \mathcal{A}_{ijkl}\,{\rm d}x_i\wedge {\rm d}x_j\wedge {\rm d}x_k\wedge {\rm d}x_l\  , \] 
with 
\begin{align*} 
\mathcal{A}_{ijkl}=& \tfrac{1}{2}{\rm tr}\left\{ \Gamma_{l;i,j} J_i\Gamma_{k;j,i} J_j - \Gamma_{k;i,j} J_i\Gamma_{l;j,i} J_j \right. \\ 
&\qquad +  \Gamma_{l;k,i} J_k\Gamma_{j;i,k} J_i - \Gamma_{j;k,i} J_k\Gamma_{l;i,k} J_i  \\ 
&\qquad \left. \Gamma_{l;j,k} J_j\Gamma_{i;k,j} J_k - \Gamma_{i;j,k} J_j\Gamma_{l;k,j} J_k\pm {\rm cycl}\ \ (ijkl) \right\}\ , 
\end{align*} 
where the cyclic permutation over the indices $(i,j,k,l)$ is done with alternating signs of the six terms inside the bracket, 
resulting in 24 terms in total. Here the quantities $\Gamma$ are given by 
\[ \Gamma_{i;j.k}=\partial_iG_{jk}-G_{ik}J_i G_{ji}\   , \] 
and hence we have a double zero expansion of ${\rm d}\mathsf{L}$ implying that the generalised Euler-Lagrange equations 
arising from $\delta{\rm d}\mathsf{L}=0$ for all $G_{ij}$ varied independently (for different indices) gives rise to the entire 
system of matrix Darboux equations to yield the critical point of the action 
\[ S[G_{\cdot,\cdot}(\boldsymbol{x});\mathcal{V}]=\int_{\mathcal{V}} \mathsf{L} \  , \] 
as a functional of all the matrix fields $G_{\cdot,\cdot}$ as well as of the hypersurfaces $\mathcal{V}$ in the space of independent 
variables. As a consequence of the double-zero expansion form we have ${\rm d}\mathsf{L}=0$ for the fields $G$ obeying the set of 
EL equations, and hence the action is independent of the choice of hypersurface for those critical fields.   
\end{proof}

More or less simultaneously to our paper \cite{NijMaill}, and independently, Bogdanov and Manakov 
investigated a (2+1)-dimensional Lagrangian matrix system, cf. \cite{BogdMan}.  
In retrospect both systems are very similar and originate from the 
consideration of nonlocal inverse problems, either through Direct Linearisation in the case of 
\cite{NijMaill}, in a framework also exploited for the case of three-dimensional matrix lattice equations, 
cf. \cite{Nij85LMP,NC90}, or using nonlocal $\bar{\partial}$ in the case of 
\cite{BogdMan}. 
Like in the scalar case of section 2, these Lagrangians can be seen as components of a Lagrangian 
3-form, and in a precise sense they generate the entire hierarchy of matrix KP equations (I will not 
dwell on that aspect in the present note). 

To be more precise let us first, in the notation of \cite{NijMaill}, specify the 
matrix (or, in the parlance of the last decade, the non-Abelian or non-commutative) KP 
structure. Note, that one of the first papers that addressed the matrix KP system, 
from an inverse scattering point of view, was \cite{Konop82}. The main set of equations, in fact the matrix generalisation of \eqref{eq:Srels}, is the family of relations given by  
\begin{equation}\label{eq:matHeq} 
\partial_k^J H_{ab}=\frac{J}{k-a}H_{ab}-H_{ab}\frac{J}{k-b} +H_{ak}JH_{kb}\  , 
\end{equation} 
where $k\neq a\neq b\neq k $ are complex valued parameters and the derivatives 
$\partial_k^J$ is with respect to some Miwa type variables $\xi^J_k$ characterised by 
the constant matrix $J$ as well as the label $k$ which here is complex parameter. 
The family of equations \eqref{eq:matHeq} is multidimensionally 
consistent for different values of $k$ and commuting sets of matrices $J$, as can be readily verified. 

A Lagrangian for the set of equations \eqref{eq:matHeq} is given by 
\begin{align}\label{eq:HLagr} 
\mathcal{L}_{klm} & 
=\tfrac{1}{2}{\rm tr}\left\{ H_{ml}\widetilde{J}(\partial^J_k H_{lm})\widehat{J} 
- (\partial_k^J H_{ml})\widetilde{J} H_{lm} \widehat{J} \right.  \nonumber \\ 
& + H_{km}\widehat{J}(\partial^{\widetilde{J}}_l H_{mk})J 
- (\partial_l^{\widetilde{J}} H_{km})\widehat{J} H_{mk} J \nonumber \\ 
& \left. + H_{lk}J(\partial^{\widehat{J}}_m H_{kl})\widetilde{J} 
- (\partial_m^{\widehat{J}} H_{lk})J H_{kl} \widetilde{J} \right\}  \nonumber \\  
& + {\rm tr} \left\{ H_{ml}\widetilde{J}H_{lm}\frac{J\widehat{J}}{k-m} 
-H_{ml} \frac{\widetilde{J}J}{k-l} H_{lm} \widehat{J}  \right. \nonumber \\  
& + H_{km}\widehat{J}H_{mk}\frac{\widetilde{J}J}{l-k} 
-H_{km} \frac{\widehat{J}\widetilde{J}}{l-m} H_{mk} J \nonumber \\ 
& \left. + H_{lk}JH_{kl}\frac{\widehat{J}\widetilde{J}}{m-l} 
-H_{lk} \frac{\widehat{J}J}{m-k} H_{kl} \widetilde{J}\right\} \nonumber \\ 
& + {\rm tr}\left\{ H_{lm}\widehat{J} H_{mk} J H_{kl} \widetilde{J} 
- H_{ml}\widetilde{J} H_{lk} J H_{km} \widehat{J} 
  \right\}\ ,  
\end{align}    
which essentially is equivalent to the Lagrangian of \cite{BogdMan}. The variational 
equations 
\[ \frac{\delta\mathcal{L}_{klm}}{\delta H_{ml}^T}=0 \quad \Rightarrow \quad 
\partial_k^J H_{lm}=\frac{J}{k-l}H_{lm}-H_{lm}\frac{J}{k-m} +H_{lk}JH_{km}\  , 
 \] 
 and similarly the other equations with $k,l,m$ and $J,\widetilde{J}, \widehat{J}$ 
respectively, all permuted, follow from this Lagrangian. By expanding the Miwa variables 
 we can derive Lagrangians for the matrix KP hierarchy (examples of matrix KP hierarchy 
 equations arising from the analogous Lagrange structure were 
 provided in \cite{BogdMan}). The main new insight provided here, 
 and which is a direct consequence, in fact a specification, of Theorem 5.1,   
is that this Lagrangian structure can be extended to a Lagrangian 3-form structure for the 
matrix KP hierarchy in (matrix) Miwa variables\footnote{Since integrable matrix hierarchies comprise 
not a single sequence of higher time-flows, but several families, each generated by a  
zeroth order time-flow associated with a constant matrix $J$, this matrix serves as the 
label for the corresponding sequence of higher times $t^J_j$, cf. e.g. \cite{Dickey-book}, 
and associated Miwa type variables $\xi_p^J$ can be defined by `compounding' those hierarchies in 
the sense of \cite{Nij88}, i.e. constructing weighted sums of higher time derivatives 
as in \eqref{eq:Miwa}. }, provided by 
\newcommand{\tvarphi}{{}^t\!\varphi}
\begin{align*} \mathsf{L} = 
& \mathcal{L}_{klm} {\rm d}\xi^J_k\wedge {\rm d}\xi_l^{\widetilde{J}}\wedge 
{\rm d}\xi_m^{\widehat{J}}+ 
\mathcal{L}_{lmn} {\rm d}\xi^{\widetilde{J}}_l\wedge {\rm d}\xi_m^{\widehat{J}}\wedge 
{\rm d}\xi_n^{\overline{J}} \nonumber \\ 
& + \mathcal{L}_{mnk} {\rm d}\xi^{\widehat{J}}_m\wedge {\rm d}\xi_n^{\overline{J}}\wedge 
{\rm d}\xi_k^{J}
+\mathcal{L}_{klm} {\rm d}\xi^J_k\wedge {\rm d}\xi_l^{\widetilde{J}}\wedge 
{\rm d}\xi_m^{\widehat{J}}
\end{align*}
(which can be readily extended to a multi-sum involving more variables of the type $x_k^J$ 
with different labels and different matrices $J$). As a conclusion this provides 
the proper variational structure of the matrix KP hierarchy in its generating form.  
This is direct consequence of Theorem 5.1, where the correspondence between the matrices  
$G_{kl}$ and the matrices $H_{kl}$ is obtained by introducing matrix analogues of the 
plane-wave factors $\rho_k$ and $\sigma_{k'}$, given by nonsingular $N\times N$ matrices 
$\varphi^0_k$ and $\tvarphi^0_l$  obeying
\[ \partial_m^J\varphi^0_k=\frac{J}{m-k}\varphi^0\ , \quad 
\partial_m^J \tvarphi^0_l= -\tvarphi^0_l\frac{J}{m-l}\  ,  \]
(where the superscript $\phantom{}^0$ denotes the aspect that these are 'free' solutions 
of the underlying linear system), and setting 
$$ G_{kl}=\tvarphi^0_l H_{lk}\varphi_k^0\  , \quad {\rm and}\quad J_m=(\tvarphi^0_m)^{-1} J (\varphi^0_m)^{-1}\  . $$      
As a consequence, relying on Theorem 5.1, the Lagrangian \eqref{eq:HLagr} form the 
components of a Lagrangian 3-form whose generalised EL equations provide the system of 
equations \eqref{eq:matHeq}. This is essentially the generating set of equations for the 
matrix KP hierarchy. 

\section{Discussion}\label{S:Concl}

The results in this paper generalise in an essential way those of \cite{SNC21} where the 
multiform structure of the KP hierarchy was established in the conventional presentation 
in terms of pseudo-differential operators. In this paper we consider the KP hierarchy 
from the point of view of \textit{generating PDEs}, namely through their representation 
in terms of Miwa variables. This has the advantage that the structure becomes much more 
covariant. Thus the KP hierarchy is being treated as multi-parameter family of equations 
in the sense of what we called a \emph{generating PDE}, i.e. a PDE in terms of Miwa type 
variables, which by expansion in powers of the parameters lead to the conventional hierarchy 
of KP equations in the multi-time form (in the cases of (1+1)-dimensional hierarchies 
these generating PDEs are obtained from the conventional enumerative hierarchies, by 
a process of `compounding', cf \cite{Nij88}). A connection between Lagrangian multiforms 
in this parameter-family representation and the classical $r$-matrix was recently put forward in 
\cite{CaudStopp}. Lifting those results to the case of the (2+1)-dimensional KP hierarchy 
could provide a novel route to the quantisation of the KP system. Identifying a classical (and 
possible in due course a quantum) $R$-matrix for the KP system would form a major step 
towards both a canonical as well as path integral route towards its quantisation. 

In this context it is worth mentioning another connection. In the Direct Linearising 
Transform (DLT) approach to the KP system (discrete as well as continuous), cf. 
\cite{NCWQ84,Nij85LMP,NC90,FN16} the invariance under integral transforms with a kernel 
$G_{kk'}$ is the key element of the construction. This kernel 
is a path-independent line-integral in the space of independent variables of the system, 
constructed from of a closed (on solutions of the equation of motion) 1-form 
constructed from the Lax multiplet eigenfunctions. The kernel $G_{kk'}$ solves 
the following class of generalised Darboux systems
\begin{equation}\label{eq:Geq} 
 \partial_i G_{kk'}=  \iint_{D_i} G_{lk'}\,{\rm d}\zeta_i(l,l')\, G_{kl'}\  , 
\quad i\in I,  
\end{equation}  
where the integration is over a set (labelled by $I$) of domains $D_i\in \mathbb{C}\times\mathbb{C}$ in some 
spectral type variables $l$ and $l'$ over a set of matrix valued measures ${\rm d}\zeta_i(l,l')$ 
in that domain. The independent variables are assumed to be characterised by the 
integration data: $x_i=x(\zeta_i,D_i)$ and $\partial_i=\partial/\partial x_i$. 
Notable is the dual role played by $G_{kk'}$, on the one hand as the integral kernel of an integral 
transform, on the other hand as a solution of a parameter-family of nonlinear equations 
of Darboux type, which can be reconstructed from the quantities $H_{k,k'}$ of the previous section.  
Most important in the present context is the observation that this 
general system can be endowed with a Lagrangian 3-form structure very similar 
to the ones described in the previous section, namely given by 
Lagrangian components 
\begin{align}\label{eq:dZLagr}
\mathcal{L}_{ijk}=& \tfrac{1}{2}\iint_{D_i}\iint_{D_j} {\rm tr}\left\{ 
G_{l,k'} \,{\rm d}\zeta_i(l,l')\, (\partial_k G_{k,l'})\, {\rm d}\zeta_j(k,k') \right. \nonumber\\
& -\left. (\partial_k G_{l,k'})\, {\rm d}\zeta_i(l,l')\, G_{k,l'}\, {\rm d}\zeta_j(k,k')
+ {\rm cycl}(ijk) \right\} \nonumber \\ 
& + \iint_{D_i} \iint_{D_j} \iint_{D_k} {\rm tr}\left\{ 
G_{l,k'}\,{\rm d}\zeta_i(l,l')\,G_{m,l'}\,{\rm d}\zeta_j(m,m')\,G_{k,m'}\,
{\rm d}\zeta_k(k,k') \right. \nonumber \\ 
& \qquad \qquad \left. -  G_{l,k'}\,{\rm d}\zeta_i(l,l')\,G_{m,l'}\,{\rm d}\zeta_k(m,m')\,G_{k,m'}\,
{\rm d}\zeta_j(k,k')  \right\} \  . 
\end{align}
It can be proven by similar computations, and under some generous assumptions 
on the integrations in the formula, that analogous statements to the ones  
in the previous sections, that the Lagrangian 3-form with components given by 
\eqref{eq:dZLagr} possesses a Lagrangian multiform structure. This forms arguably the 
most general multiform structure so far considered in the theory. Note also that the 
corresponding action functional $S[G_{\cdot,\cdot}(\boldsymbol{x});\mathcal{V}]$ 
where as before $\mathcal{V}$ is an arbitrary 3-dimensional hypersurface in the space of 
independent variables $\boldsymbol{x}=(\{ x_i, i\in I\})$, 
shows some resemblance some action functionals associated with the 
Chern-Simons theory in topological field theory, but this connection still remains 
to be explored.  

\subsection*{Acknowledgements}

The author has benefited from many discussions with V. Caudrelier, L. Peng, D. Sleigh and M. Vermeeren 
on various issues regarding Lagrangian multiform theory. 
He is supported by EPSRC grant EP/W007290/1

\subsection*{Data availability statement} 
Data sharing is not applicable to this article as no datasets were generated or analysed during the current study.

\subsection*{Conflict of interest statement} 
The corresponding author states that there is no conflict of interest.


\small
\begin{thebibliography}{99}

 
\bibitem{BobSur15}
A.I. Bobenko and Yu.B. Suris. Discrete pluriharmonic functions as solutions of linear pluri-Lagrangian systems. \emph{Communications in Mathemathical Physics}, 2015, \textbf{336} 199-215.
\bibitem{BogdKonopel95} 
L.V. Bogdanov and B.G. Konopelchenko, Lattice and $q$-difference Darboux-Zakharov-Manakov 
systems via $\bar{{\partial}}$-dressing method, \emph{J.Phys. A: Math. Gen.} {\bf 28} no. 5 (1995) 
L173--L178.    
\bibitem{BogdKonopel} 
L.V. Bogdanov and B.G. Konopelchenko, Analytic-bilinear approach to integrable 
hierarchies. I Generalized KP hierarchy, \emph{Journ. Math. Phys.} {\bf 39} no. 9 (1998) 
4683--4700.  
\bibitem{BogdMan}
L.V. Bogdanov and S.V. Manakov, The non-local $\bar{\partial}$ problem and (2+1)-dimensional 
soliton equations, J. Phys, A: Math. Gen. {\bf 21} (1988) L537-L544. 

\bibitem{BolPetSur2013b}
R. Boll, M. Petrera and Yu.B. Suris. What is integrability of discrete variational systems? \emph{Proceedings of the Royal 
Society A}, 2014, \textbf{470}: 20130550 (published 11 December 2013).
\bibitem{BolPetSur2015} 
R. Boll, M. Petrera and Yu.B. Suris. On the variational interpretation of the discrete KP equation. In: \emph{Advances in Discrete Differential Geometry}, Ed. 
A.I. Bobenko, (Springer Verlag, 2016) pp. 379--405.  
  

\bibitem{CNSV2022}
V. Caudrelier, F.W. Nijhoff, D. Sleigh and M. Vermeeren, 
Lagrangian multiforms on Lie groups and noncommuting flows, ArXiv:2204.09663. 
\bibitem{CaudStopp} 
V. Caudrelier, M. Stoppato and B. Vicedo, Classical Yang-Baxter equation, Lagrangian multiforms and 
ultralocal integrable hierarchies, arXiv: 2201.08286. 
\bibitem{CDS97}
J. Cie\'sli\'nski, A. Doliwa and P.M. Santini, The integrable discrete analogue of orthogonal 
coordinate systems are multidimensional circular lattices, \emph{Phys. Lett. A} {\bf 235} (1997) 480--488. 
\bibitem{Darboux} 
G. Darboux, M\'emoire sur la th\'eorie des coordonn\'ees curvilignes, et des syst\`emes 
orthogonaux. \textit{Ann. Sci. de l'\'ENS}, 2eme s\'erie, tome {\bf 7} (1878) 275--348.  
\bibitem{DickeyKP} 
L.A. Dickey, On Hamiltonian and Lagrangian formalisms for the KP-hierarchy of integrable
equations. \textit{Ann. New York Acad. Sci.} {\bf 491}, no. 1 (1987) 131–48. 
\bibitem{Dickey-book} 
L.A. Dickey, Soliton Equations and Hamiltonian Systems, 2nd ed. World Scientific, 2003. 
\bibitem{Doliwa}
A. Doliwa, On $\tau$-function of conjugate nets, \textit{J. of Nonl. Math. Phys.} {\bf 12} (2004) 
244--252.   
\bibitem{DS97}
A. Doliwa and P.M. Santini, 
    Multidimensional quadrilateral lattices are integrable.
    \textit{Phys. Lett. A} \textbf{233}(1997) 365--372.
\bibitem{Doliwa-etal}
A. Doliwa, M. Ma\~nas, L. Martinez-Alonso, E. Medina and P.M. Santini, Charged free 
fermions, vertex operators and the classical theory of conjugate nets, 
\emph{J Phys A: Math. Gen.} {\bf 32} (1999) 1197–1216. 
\bibitem{DN91}
I.Ya. Dorfman and F.W. Nijhoff, On a (2+1)-dimensional version of the Krichever-Novikov equation, 
\emph{Phys. Lett. A} {\bf 157} (1991) 107--112. 
\bibitem{FN16}
W. Fu and F.W. Nijhoff, 
    Direct linearizing transform for three-dimensional discrete integrable systems: the lattice AKP, BKP and CKP equations.
    \emph{Proc. R. Soc. A} \textbf{473} (2016) 20160915.
\bibitem{FN17}
    W. Fu and F.W. Nijhoff, 
    Linear integral equations, infinite matrices and soliton hierarchies.
    \emph{J. Math. Phys.} {\bf 59} (2018) 071101. 
\bibitem{HJN16}
    J. Hietarinta, N. Joshi and F.W. Nijhoff, 
    \textit{Discrete Systems and Integrability}. \textit{Cambridge Texts in Applied Mathematics} 
    (Cambridge University Press, 2016).
\bibitem{Konop82} 
B.G. Konopelchenko, On the general structure of nonlinear evolution equations and their 
B\"acklund transformations connected with the matrix non-stationary Schr\"odinger 
spectral problem, \emph{J. Phys. A: Math.Gen.} {\bf 15} (1982) 3425–-3437. 
\bibitem{Konop}
B.G. Konopelchenko, Introduction to Multidimensional Integrable 
Equations. The Inverse Spectral Transform in 2+1 Dimensions, 
(Plenum Press, New York, 1992). 
\bibitem{Lame}
G.Lam\'e, Le\c{c}ons sur les coordonn\'ees curvilignes et leurs diverses applications, (Mallet–Bachalier, Paris,
1859).     
\bibitem{LobbNij2009}
S. Lobb and F.W. Nijhoff, Lagrangian multiforms and multidimensional consistency, \textit{J. Phys. A: Math. 
Theor.} {\bf 42} (2009) 454013. 
\bibitem{LNQ}
S. Lobb, F.W. Nijhoff and G.R.W. Quispel, Lagrangian multiform structure for the lattice KP system, 
\textit{J. Phys. A: Math. Theor.} {\bf 42} (2009) 472002.   
\bibitem{LobbNij2018} 
S. Lobb and F.W. Nijhoff, A variational principle for discrete integrable systems, \textit{SIGMA} {\bf 14} 
(2018) 041.  
\bibitem{MDS97} 
M. Ma\~nas, A. Doliwa and P.M. Santini, Darboux transformations for multidimensional 
quadrilateral lattices. I. \emph{Phys. Lett. A} {\bf 232} (1997) 99--105. 
\bibitem{MartinezKonop}
L. Martinez-Alonso and B. Konopelchenko, The KP Hierarchy in Miwa coordinates, \textit{Phys. Lett. A} {\bf 258} 
(1999) 272--278. 
\bibitem{Miw82}
T. Miwa, 
    On Hirota's difference equations.
    \textit{Proc. Japan Acad.} (1982) \textbf{58A}, 9--12.
\bibitem{Nij85LMP}
F.W. Nijhoff, 
Theory of integrable three-dimensional lattices, 
\textit{Lett. Math. Phys.}, \textbf{9} (1985) 235--241.  
\bibitem{Nij-tau85}   
    F.W. Nijhoff, 
The Direct Linearizing Transform for the $\tau$-function in Three-dimensional Lattice Equations, 
\textit{Phys. Lett}. {\bf 110A} (1985) 10--14.
\bibitem{Nij88} 
F.W. Nijhoff, Linear integral transformations and hierarchies of integrable nonlinear evolution 
equations, Physica {\bf D31} (1988) 339--388. 
\bibitem{Nijtbp} 
F.W. Nijhoff, Lagrangian multiform structure of discrete and semi-discrete KP systems, in preparation. 
\bibitem{NCWQ84} 
F.W. Nijhoff, H.W. Capel, G.L. Wiersma and G.R.W. Quispel, B\"acklund transformations and 
three-dimensional lattice equations,  
\textit{Phys. Lett.}  {\bf 105A} (1984) 267--272.
\bibitem{NC90} 
F.W. Nijhoff and H.W. Capel, The direct linearization approach to hierarchies of integrable
PDE’s in 2 + 1 dimensions. I. Lattice equations and the differential–difference hierarchies.
\textit{Inverse Probl.} \textbf{6} (1990) 567--590.   
\bibitem{NHJ}
F.W. Nijhoff, A. Hone and N. Joshi, On a Schwarzian PDE associated with the KdV hierarchy, 
\textit{Phys. Lett. A} {\bf 267} (2000) 147--156. 
\bibitem{NijMaill}
F.W. Nijhoff and J.-M. Maillet, Algebraic Structure of Integrable Systems in $D=2+1$ and Routes 
towards Multidimensional Integrability, in: \emph{Nonlinear Evolutions}, Proceedings of 
the IVth NEEDS Conference, Ed. J.J.P. L\'eon, (World Scientific, Signapore, 1988);  
Preprint PAR-LPTHE 87-45, September 1987. 
 
  
\bibitem{SAF1984} 
P.M. Santini, M.J. Ablowitz and A.S. Fokas, The direct linearization of a class of nonlinear evolution equations, 
\textit{J. Math. Phys.}, \textbf{25} (1984) 2614--2619.  
\bibitem{Sato} 
M. Sato, Soliton equations as dynamical systems on infinite dimensional Grassmann manifolds, 
\textit{RIMS K\^oky\^uroku} {\bf 439} (1981) 30--46. 
\bibitem{SNC21} 
D. Sleigh, F.W. Nijhoff and V. Caudrelier, Lagrangian multiforms for Kadomtsev-Petviashvili (KP) and the 
Gel'fand-Dickey hierarchy. \textit{Int. Math. Res. Notices} Vol. 2021, 1--41.  
\bibitem{SNC19} 
D. Sleigh, F.W. Nijhoff and V. Caudrelier, A variational approach to Lax representations, \textit{Journ. Geom. Phys.} 
{\bf 142} (2019) 66--79.  
\bibitem{SNC20} 
D. Sleigh, F.W. Nijhoff and V. Caudrelier, Variational symmetries and Lagrangian multiforms, \textit{Lett. Math. Phys.} 
{\bf 110} no. 4 (2020) 805--826.  
\bibitem{Sleigh-thesis}
D. Sleigh, The Lagrangian multiform approach to integrable systems, University of Leeds, PhD thesis (2020).  
 
\bibitem{Suris2012} 
Yu. B. Suris, Variational formulation of commuting Hamiltonian flows: multi-time Lagrangian 1-forms.  
\emph{J. Geometric Mechanics}, 2013, \textbf{5} N0. 3, pp. 365--379. 
\bibitem{Suris2016} 
Y. Suris, Variational symmetries and pluri-Lagrangian systems, in \textit{Dynamical Systems, 
Number Theory and Applications}, (Springer Verlag, 2016) 255--266
\bibitem{SurVerm} 
Yu. Suris and M. Vermeeren, On the Lagrangian structure of integrable hierarchies, in: \textit{Advances in Discrete Differential Geometry}, ed. A. Bobenko, (Springer Verlag, Berlin), pp. 347--378. 
\bibitem{Verm}
M. Vermeeren, Continuum limits of variational systems. Technische Universit\"at Berlin, PhD thesis (2018). 
  
\bibitem{TN}
A. Tongas and F.W. Nijhoff, Generalized hyperbolic Ernst equations for an Einstein-Maxwell-Weyl field, 
\textit{J. Phys A: Math. Gen.} {\bf 38} (2005) 895--906. 
 
\bibitem{Y-KLN} 
S. Yoo-Kong, S. Lobb and F.W. Nijhoff, Discrete-time Calogero-Moser systems and Lagrangian 1-form structure. 
\textit{J. Phys. A: Math. Theor.} {\bf 44} (2011) 365203.  
  
\bibitem{ZM85}
V.E. Zakharov and S.V. Manakov, 
    Construction of higher-dimensional nonlinear integrable systems and of their solutions.
    \textit{Funct. Anal. Appl.} (1985) \textbf{19}, 89--101.
 
\bibitem{Zakharov96}
V.E. Zakharov, On integrability of thev equations describing $N$-orthogonal curvilinear 
coordinate systems and Hamiltonian integrable systems of hydrodynamic type. Part I. Integration of the 
Lam\'e equations. \textit{Duke Math. J.} {\bf 94} (1996) 103--139.  
  



\end{thebibliography}
\end{document}